\documentclass[11pt,a4paper]{article}

\usepackage{amsmath,amssymb,amsthm}
\usepackage{mathtools}
\usepackage{geometry}
\usepackage[hidelinks]{hyperref}
\usepackage{booktabs}
\usepackage{enumitem}
\usepackage{microtype}
\usepackage{graphicx}

\newtheorem{theorem}{Theorem}
\newtheorem{proposition}{Proposition}

\newtheorem{corollary}{Corollary}

\newtheorem{remark}{Remark}

\newcommand{\calC}{\mathcal{C}}
\newcommand{\calH}{\mathcal{H}}
\newcommand{\calX}{\mathcal{X}}
\newcommand{\calY}{\mathcal{Y}}
\newcommand{\calT}{\mathcal{T}}
\newcommand{\E}{\mathbb{E}}
\newcommand{\IQ}{I_Q(X;Y)}
\newcommand{\EFA}{E_{\mbox{\tiny FA}}}
\newcommand{\EMD}{E_{\mbox{\tiny MD}}}
\newcommand{\Dc}{D_{\mbox{\tiny c}}(Q_{XY})}
\newcommand{\Dm}{D_{\mbox{\tiny m}}(Q_Y)}

\begin{document}
\thispagestyle{empty}
\title{\textbf{Soft Covering Through the Lens of Hypothesis Testing}}
\author{Neri Merhav}

\date{\today}
\maketitle

\begin{center}
The Viterbi Faculty of Electrical and Computer Engineering\\
Technion - Israel Institute of Technology \\
Technion City, Haifa 3200003, ISRAEL \\
E--mail: {\tt merhav@technion.ac.il}\\
\end{center}

\vspace{1.5\baselineskip}
\setlength{\baselineskip}{1.5\baselineskip}

\begin{abstract}
The soft covering lemma asserts that a random codebook whose rate $R$
exceeds the input-output mutual information, $I(X;Y)$, of a given discrete memoryless channel (DMC),
causes the output mixture distribution to be statistically indistinguishable
from the i.i.d.\ output distribution.
We study this phenomenon through the lens of Neyman--Pearson hypothesis
testing: given a channel output sequence $y^n$, can one decide whether it was
produced when the channel was driven by a random codeword, or generated
independently from the output marginal?
We derive exact exponential decay rates for the jointly averaged
false-alarm (FA) probability $\alpha_n(\tau,R)$ and
missed-detection (MD) probability $\beta_n(\tau,R)$,
as functions of the decision threshold $\tau$ and the codebook rate $R$.
The derived single-letter formulas of the exponents $\EFA(\tau,R)=-\lim_{n\to\infty}\frac{1}{n}\ln\alpha_n(\tau,R)$
and $\EMD(\tau,R)=-\lim_{n\to\infty}\frac{1}{n}\ln\beta_n(\tau,R)$
are tight in the random coding sense.
The analysis reveals a rich phase structure.
For $R < I(X;Y)$, there is a genuine exponential tradeoff between the
two error types over the interval $\tau \in (0, I(X;Y)-R)$.
At $R = I(X;Y)$, this tradeoff interval collapses to the single point
$\tau = 0$, where both error exponents simultaneously vanish, a fact which
manifests the soft covering phenomenon in the Neyman--Pearson sense.
For $R > I(X;Y)$, the same instantaneous collapse persists at $\tau = 0$;
moreover, for every $\tau$ at least one exponent is zero:
the FA exponent is zero for $\tau \le 0$ (FA probability does not
decay exponentially), and the MD exponent is zero
for $\tau \ge 0$ (and finite, channel-specific for $\tau<0$; see Remark~\ref{rem:jump}).
There is no interval of $\tau$ where both exponents are simultaneously positive.
A sharp phase transition in the MD exponent occurs at $\tau^* = [I(X;Y)-R]_+$
for all rates.\\

\end{abstract}

\section{Introduction}
\label{sec:intro}

The soft covering lemma, introduced by Han and Verd\'u~\cite{HanVerdu93},
is a cornerstone of information theory.
It asserts that a random codebook $\calC$ of rate $R > I(X;Y)$ causes
the channel output mixture distribution $P_{Y^n|\calC}$ to become statistically
indistinguishable from the i.i.d.\ output distribution $P_Y^{\otimes n}$,
in the sense that their total variation distance vanishes exponentially.
This result underpins key results in channel coding, wiretap secrecy,
common randomness generation, and random-binning arguments.
Despite its importance, the classical soft covering results address
only the regime $R > I(X;Y)$.
The regime $R < I(X;Y)$ --- where the codebook rate is below the
mutual information --- has received comparatively little attention,
and it turns out that it is considerably the more challenging and structurally richer regime.

In this paper, we study the soft covering phenomenon across \emph{all}
rates $R \geq 0$ by formulating it as a Neyman--Pearson hypothesis
testing problem:
Given an observed output sequence $y^n$, we ask: was it produced by
the channel $W$ driven by a randomly chosen codeword from $\calC$
(hypothesis $\calH_1$), or was it drawn independently from the marginal
$P_Y^{\otimes n}$ (hypothesis $\calH_0$)?
The Neyman--Pearson test, based on the log-likelihood ratio (LLR) $\Lambda(y^n)$,
decides in favor of $\calH_1$ when $\Lambda(y^n) \geq \tau$, where $\tau$
is a threshold tuned to obtain a prescribed tolerable FA probability.
The two error events are then:
false alarm (FA): deciding in favor of $\calH_1$ when $\calH_0$ is true, and
missed detection (MD): deciding in favor $\calH_0$ when $\calH_1$ is true.
We derive the exact exponential decay rates $\EFA(\tau,R)$ and
$\EMD(\tau,R)$ of the jointly averaged FA and MD error probabilities
$\alpha_n(\tau,R)$ and $\beta_n(\tau,R)$, to be defined formally in
Section~\ref{sec:notation}.

The analysis reveals a considerably surprising phase structure in the $(\tau,R)$ plane,
governed by two critical thresholds: $\tau=0$ and $\tau^*=
\max\{0,I(X;Y)-R\}$.
For $\tau<0$, the FA exponent vanishes (unless the channel has some
singularity like the Z-channel) and
the MD exponent is positive and channel-specific (see Remark~\ref{rem:jump}). At $\tau=0$, the picture depends critically on $R$.
The FA exponent is always zero at $\tau=0$ (with the above digression for
singular channels), for every $R>0$.
The MD exponent at $\tau=0$ undergoes a phase transition at $R=I(X;Y)$:
it is strictly positive for $R<I(X;Y)$, and vanishes for $R\ge I(X;Y)$.
The case $R=I(X;Y)$, $\tau=0$, where \emph{both} exponents simultaneously
vanish, is the Neyman--Pearson exponent formulation of the soft covering phenomenon:
at exactly the soft covering rate, neither error decays exponentially at threshold $\tau=0$.
For $\tau>0$, both error exponents can be positive, and the picture depends on $R$.
For $R < I(X;Y)$,
there is a genuine tradeoff interval $\tau\in(0,\,I(X;Y)-R)$
where both error exponents are simultaneously positive.
Beyond this interval (for $\tau\ge I(X;Y)-R$), the MD exponent drops to zero.
The tradeoff interval has width $I(X;Y)-R$, which shrinks to zero as
$R\nearrow I(X;Y)$.
For $R \ge I(X;Y)$
at least one of the two exponents is zero no matter what the value of $\tau$
may be, in other words,
there is no interval where both exponents are simultaneously positive.
Specifically: for $\tau\le 0$, the FA exponent is zero (FA probability
does not decay exponentially); for $\tau\ge 0$, the MD exponent is zero (the codeword
output is statistically easy to detect).
The FA exponent is strictly positive for $\tau>0$ and grows with $\tau$,
but over that same range the MD exponent is identically zero.
The two zero-regions cover the entire real line, overlapping only at
$\tau=0$ where both exponents vanish simultaneously.

A structural asymmetry between the two exponents emerges from the formulas
(stated precisely in Theorem~\ref{thm:main}).
The FA exponent penalizes both the deviation of the channel output empirical distribution
from the i.i.d.\ marginal, and the rate surplus of a codeword's mutual information
over $R$, reflecting the rarity of the event that a noise sequence
looks like a codeword output.
The MD exponent penalizes only the deviation of the empirical channel
from the true channel $W$, with no explicit rate term, reflecting the rarity
of the event that a codeword output goes undetected.
The proofs are based on large deviations properties concerning type-class enumerators
\cite[Chapter 4]{MW25}. 
These enumerators are binomial random variables with
exponentially many trials and exponentially decaying probabilities of success.



A few words on earlier related work are in order.

Han and Verd\'u~\cite{HanVerdu93} introduced channel resolvability and
established that the minimum rate $R$ for which the mixture distribution
$P_{Y^n|\calC}$ can approximate $P_Y^{\otimes n}$ (in total variation or
normalized KL divergence) equals $I(X;Y)$; this is the first-order result
that forms the foundation of all subsequent work.
The soft-covering lemma first appeared as Theorem~6.3 of
Wyner~\cite{Wyner75}, where it serves as a technical tool
for the achievability proof of the common information theorem;
it was later recognized as a central technique in wiretap secrecy,
identification coding, and channel synthesis, and made the
subject of systematic study by Han and Verd\'u~\cite{HanVerdu93}
under the name of channel resolvability.
A substantial body of later work derives the exact \emph{exponential} rate of
convergence of $P_{Y^n|\calC}$ to $P_Y^{\otimes n}$ for $R>I(X;Y)$,
under various distance measures.
Hayashi~\cite{Hayashi06} obtained a lower bound to the exponent under KL divergence.
Parizi {\em et al.} \cite{PTM17} derived the exact exponent of th KL
divergence with application to the wiretap
channel. Yu and Tan~\cite{YuTan19} characterized the exact exponent under
R\'enyi divergence of order $\alpha\in[0,2]$ (which includes the
Kullback-Leibler (KL) divergence as the limiting case $\alpha\to 1$) for
i.i.d.\ random codes.
Yagli and Cuff~\cite{YagliCuff19} established the exact exponent under
total variation distance.
Recently, Li {\em et al.} \cite{LLY26} derived a strong-converse exponent under the
KL divergence. One a somewhat different research route,
Cuff~\cite{CuffStrong15,Cuff16} moved beyond the expected-value analysis
to show that soft covering holds with probability doubly exponentially
close to unity over the random codebook, enabling applications via the
union bound. Cuff~\cite{Cuff13} also developed the theory of distributed
channel synthesis, which relies on and strengthens the soft-covering lemma.

As mentioned above, Yu and Tan~\cite{YuTan19} characterized the exact
exponential decay of the R\'enyi divergence of order 
$\alpha\in(0,2)$ and $R>I(X;Y)$.
There is a precise connection to hypothesis testing:
the R\'enyi divergence of order $\alpha$ evaluated at the
optimal $\alpha\in(0,1)$ corresponds to the \emph{Chernoff exponent}
(symmetric Bayesian error exponent)
of testing $P_Y^{\otimes n}$ against $P_{Y^n|\calC}$.
In Neyman--Pearson terms, this is the exponent achieved at the \emph{specific threshold}
$\tau$ that equalizes $\EFA(\tau,R)$ and $\EMD(\tau,R)$ (the Chernoff point); it does not
give the full Neyman--Pearson operating characteristics.

The most closely related prior work is~\cite{WM14}, which considers a
model where a transmitter either sends a codeword from a random fixed-composition
codebook of rate $R$, or is silent, outputting the all-zero vector $\mathbf{0}$.
The receiver must jointly detect whether transmission occurred and, if so, decode
the message. The figures of merit are the FA probability (deciding transmission
when silent), the MD probability (deciding silent when transmitting), and
a decoding error probability.
For a fixed composition random codebook, \cite{WM14} derives the optimal
detector/decoder in an extended Neyman--Pearson sense and characterizes the exact random
coding exponents of all three error probabilities as functions of the rate $R$
and two threshold parameters $\alpha,\beta\in\mathbb{R}$.
The analysis in~\cite{WM14} is based on the type-class enumeration method, the same
fundamental tool used here.
The similarities with the present work are as follows.
Both papers formulate the problem model as a
Neyman--Pearson binary hypothesis testing ($\calH_0$: pure noise, $\calH_1$: codeword output),
both use random fixed-composition codebooks, and both derive exact random-coding
exponents via type-class enumeration. However, at the same time,
there are several differences.
The first is that in \cite{WM14}, under $\calH_0$ the transmitter output is a
repetitive symbol $0$ that designates silence (no transmission),
so under the null hypothesis, the channel output $y^n$ is a response to the
all-zero input vector, unlike the present work where $y^n$ is the channel
response to an i.i.d.\ input source, $P_X$.
The second difference is that the analysis in~\cite{WM14} is valid for all $R\ge 0$, but is
most natural for $R>0$ with a non-trivial decoding task.
The regime $R<I(X;Y)$ is not highlighted, and the soft-covering
threshold $R=I(X;Y)$ does not play a special role in~\cite{WM14}
since the hypotheses are always different
regardless of $R$.
In the present work, $R=I(X;Y)$ is the central threshold: it is
where both error exponents simultaneously vanish, characterizing the
soft covering phenomenon in Neyman--Pearson terms.
Third, in \cite{WM14} there is also a characterization of the exponent of the
decoding error probability, which has no analogue here. Finally,
in~\cite{WM14}
there is no phase transition
analogous to the one at $\tau=0$ in the present work.
Here, the behavior of $\EMD(\tau,R)$ for $\tau\le 0$
is a distinctive feature of the problem at hand.

The present paper differs from all prior work in the following respects.
\begin{enumerate}
\item \emph{Full Neyman--Pearson tradeoff for all $\tau$.}
All prior exponent results for resolvability correspond to a fixed
scalar distance measure (KL, total variation, R\'enyi), which is the
expected distance averaged over both $\calC$ and $y^n$.
We instead characterize the complete Neyman--Pearson operating characteristic:
the pair $(\EFA(\tau,R),\EMD(\tau,R))$ for every threshold $\tau$,
giving a full tradeoff curve rather than a single operating point.

\item \emph{All rates $R\ge 0$.}
Every previous exponent result for resolvability requires $R>I(X;Y)$.
We handle all $R\ge 0$, including the unexplored regime $R<I(X;Y)$
where the codebook output is more concentrated than i.i.d.\ and where
there is a genuine Neyman--Pearson tradeoff between FA and MD.

\item \emph{Proof method and exactness.}
Prior analyses apply Chernoff-type or Gallager-style bounding to the
channel output distribution, which is a mixture of exponentially many conditional
output distributions given the various input codewords.
The Chernoff bound applied to a mixture of channel outputs
is not automatically tight: prior papers establish tightness by a
separate converse argument specific to each distance measure.
In the present work, we avoid Chernoff bounds altogether.
Instead, we apply the type-class enumeration method
\cite{MW25}, which yields the \emph{exact}
exponential rate.
\end{enumerate}
A statistical-physics perspective on soft covering across all rates,
including connections to phase transitions, appears in~\cite{Merhav26}.

The outline of the remaining part of this work is as follows.
Section~\ref{sec:notation} establishes notation conventions and provides some
background.
Section~\ref{sec:main} states the main theorem, whose proof appears in
the appendix~\ref{sec:proof}.
Section~\ref{sec:phase} develops the phase structure and the connection
to soft covering.
Section~\ref{sec:numerics} illustrates the results on the Z-channel.
Finally, in Section \ref{sec:conclusion}, we summarize the main findings and
speculate on possible future research directions.

\section{Notation Conventions and Basic Background}
\label{sec:notation}

Let $\calX$ and $\calY$ denote finite input and output alphabets.
Sequences of length $n$ are written in boldface-free lowercase, i.e.,
$x^n=(x_1,\ldots,x_n)\in\calX^n$ and $y^n=(y_1,\ldots,y_n)\in\calY^n$, where
$\calX^n$ and $\calY^n$ are the $n$-th Cartesian powers of $\calX$ and
$\calY$, respectively.
The \emph{type} of a sequence
$x^n\in\calX^n$ is the empirical probability distribution
$\hat P_{x^n}(a)=\frac{1}{n}\#\{i:x_i=a\}$, $a\in\calX$.
The equivalence set of all sequences in $\calX^n$ of type $P$ is called the
\emph{type class} $\calT(P)\subseteq\calX^n$.
For a sequence pair $(x^n,y^n)\in\calX^n\times\calY^n$, the \emph{joint
type} $Q_{XY}$ is the empirical distribution
$Q_{XY}(a,b)=\frac{1}{n}\#\{i:x_i=a,\,y_i=b\}$, $a\in\calX$, $b\in\calY$.
Its marginals are denoted $Q_X$ and $Q_Y$, and the conditional type
$Q_{Y|X}$ is defined according to $Q_{XY}(a,b)=Q_X(a)Q_{Y|X}(b|a)$.
We always restrict to joint types with $Q_X=P_X$, where $P_X$ is some fixed input
distribution. Throughout, $\hat P_{y^n}$ denotes the empirical distribution
associated with $y^n$.

Information measures induced by a given probability distribution
will be subscripted by the notation of this distribution. When this
is an empirical distribution $Q_{XY}$, the subscript will be
abbreviated by $Q$, in order to avoid cumbersome notation. Thus,
$H_Q(Y)=-\sum_y Q_Y(y)\log Q_Y(y)$ is the marginal empirical entropy of an
auxiliary random vector $Y$ governed by $Q_Y$,
$H_Q(Y|X)=-\sum_{x,y}Q_{XY}(x,y)\log Q_{Y|X}(y|x)$ is the
conditional empirical entropy of $Y$ given $X$, where $(X,Y)$ are jointly governed by $Q_{XY}$, and
$I_Q(X;Y)=H_Q(Y)-H_Q(Y|X)$ is the mutual information under $Q_{XY}$.

The discrete memoryless channel (DMC) with a finite input alphabet
$\calX$ and finite output alphabet $\calY$ will be denoted by
$W:\calX\to\calY$. When $W$ is fed by an $n$-vector $x^n\in\calX^n$, the
corresponding channel output vector $y^n$ is distributed according to
\begin{equation}
W^n(y^n|x^n)=\prod_{i=1}^n W(y_i|x_i).
\end{equation}
The input distribution $P_X$ is fixed throughout, with output marginal
\begin{equation}
\label{outmarginal}
P_Y(y)=\sum_x P_X(x)W(y|x)
\end{equation}
and mutual information
\begin{equation}
I(X;Y)=\sum_{x,y}P_X(x)W(y|x)\log\frac{W(y|x)}{P_Y(y)}.
\end{equation}
We also adopt the following shorthand notations.
\begin{equation}
D_{\mbox{\tiny m}}(Q_Y)=D(Q_Y\|P_Y)=\sum_y Q_Y(y)\log\frac{Q_Y(y)}{P_Y(y)}
\end{equation}
is the KL divergence between $Q_Y$ and $P_Y$ defined in
\eqref{outmarginal}. Also,
\begin{equation}
D_{\mbox{\tiny c}}(Q_{XY})=D(Q_{Y|X}\|W|P_X)
=\sum_{x,y}P_X(x)Q_{Y|X}(y|x)\log\frac{Q_{Y|X}(y|x)}{W(y|x)}
\end{equation}
is the conditional KL divergence between $Q_{Y|X}$ and the channel $W$ with
weighting $P_X$, and
\begin{equation}
\ell(Q_{XY})=H_Q(Y|X)+D_{\mbox{\tiny c}}(Q_{XY}) 
\end{equation}
is a notation used throughout in the proofs.

The data processing inequality (DPI) of the KL divergence implies that 
\begin{equation}
\label{dataprocessing}
D_{\mbox{\tiny m}}(Q_Y)\le D_{\mbox{\tiny c}}(Q_{XY}).
\end{equation}
To see why this is true, observe that since $Q_X=P_X$,
\begin{eqnarray}
D_{\mbox{\tiny c}}(Q_{XY})&=&
D(Q_{XY}\|P_X\otimes W)-D(Q_X\|P_X)\nonumber\\
&=&D(Q_{XY}\|P_X\otimes W)-D(P_X\|P_X)\nonumber\\
&=&D(Q_{XY}\|P_X\otimes W)\nonumber\\
&\ge&D(Q_Y\|P_Y)\nonumber\\
&=&D_{\mbox{\tiny m}}(Q_Y),
\end{eqnarray}
where the inequality is the DPI of the KL divergence applied to the
marginalization map $(x,y)\mapsto y$.

For two positive sequences, $\{a_n\}_{n\ge 1}$ and $\{b_n\}_{n\ge 1}$,
we write $a_n\doteq b_n$ to mean
$\lim_{n\to\infty}\frac{1}{n}\log\frac{a_n}{b_n}=0$.
Accordingly, $a_n\doteq 0$ means that $a_n$ decays faster than exponentially (e.g.,
doubly exponentially) and $a_n\doteq 1$ tells that $a_n$ varies (grows or
decays) at a sub-exponential rate, or even tends to a strictly positive constant.
Throughout this paper, all logarithms are natural unless specified otherwise.
We also denote the positive clipping operator by $[\cdot]_+$, which is defined
by $[a]_+=\max\{a,0\}$ for any real $a$.

A random codebook $\calC=\{x^n(m)\}_{m=1}^M$, $M=e^{nR}$,
has codewords drawn independently at random under the uniform distribution
across the type class $\calT(P_X)$.
Once the codebook $\calC$ has been randomly selected, the induced output mixture is given by
\begin{equation}
  P_{Y^n|\calC}(y^n) = \frac{1}{M}\sum_{m=1}^M W^n(y^n|x^n(m)).
  \label{eq:mixture}
\end{equation}
We denote $P_Y^{\otimes n}(y^n)=\prod_{i=1}^n P_Y(y_i)$, where the
single-letter output marginal $P_Y$ is induced by $P_X$ and $W$ as in
(\ref{outmarginal}).

In this paper, we focus on the following Neyman--Pearson hypothesis testing problem:
Under hypothesis $\mathcal{H}_0$, $y^n$ is governed by $P_Y^{\otimes n}$, and under
hypothesis $\mathcal{H}_1$, $y^n$ is drawn by $P_{Y^n|\calC}$.
The log-likelihood ratio (LLR) statistic is defined as
\begin{equation}
  \Lambda(y^n) = \frac{1}{n}\ln
    \frac{P_{Y^n|\calC}(y^n)}{P_Y^{\otimes n}(y^n)}.
  \label{eq:LLR}
\end{equation}
The likelihood ratio test (LRT) with threshold $\tau$ decides in favor of
$\mathcal{H}_1$ if $\Lambda(y^n)\ge\tau$; otherwise, it accepts $\mathcal{H}_0$.

The two kinds of error probabilities associated with the LRT are as follows.
The false-alarm (FA) probability is 
\begin{equation}
\alpha_n(\tau,R) = \E_{\calC}\!\left\{
\sum_{\{y^n:~\Lambda(y^n)\ge\tau\}}
P_Y^{\otimes n}(y^n)\right\}, \label{eq:FA}\\
\end{equation}
and the missed detection (MD) probability is
\begin{equation}
\beta_n(\tau,R) = \E_{\calC}\!\left\{
\sum_{\{y^n:~\Lambda(y^n)<\tau\}}
P_{Y^n|\calC}(y^n)\right\},  \label{eq:MD}
\end{equation}
where in both (\ref{eq:FA}) and (\ref{eq:MD}), $\E_{\calC}\{\cdot\}$ denotes
expectation with respect to (w.r.t.) the randomness of the codebook $\calC$.
The corresponding error exponents are defined as
\begin{align}
  \EFA(\tau,R) &= -\lim_{n\to\infty}\frac{1}{n}\ln\alpha_n(\tau,R),
  \label{eq:EFA_def}\\
  \EMD(\tau,R) &= -\lim_{n\to\infty}\frac{1}{n}\ln\beta_n(\tau,R),
  \label{eq:EMD_def}
\end{align}
where the existence of these limits will become apparent from the derivations
to follow (Theorem~\ref{thm:main}).

The following quantity will be useful in the proofs.
Given the randomly selected codebook, $\calC$, and a channel output vector,
$y^n$, the
type-class enumerator (TCE) associated with type $Q_{XY}$ and $y^n$ is defined as
\begin{equation}
N(Q_{XY}|y^n) = \#\bigl\{m:(x^n(m),y^n)\in\calT(Q_{XY})\bigr\}.
\label{eq:TCE}
\end{equation}
Clearly, the randomness of $\calC$ induces randomness of $N(Q_{XY}|y^n)$. In
particular, due to the independent random selection, $N(Q_{XY}|y^n)$
is a binomial random variable with $M=e^{nR}$ trials and success rate 
given by the probability that a single randomly chosen codeword from $\calT(P_X)$
happens to have, together with $y^n$, the given joint type $Q_{XY}$. This
probability is
of the exponential order of $e^{-n\IQ}$. 
The unnormalized mixture under $\calH_1$ can be easily expressed in terms of the TCE's as
follows:
\begin{equation}
S(y^n)=\sum_{m=1}^M W^n(y^n|x^n(m))
=\sum_{Q_{XY}} N(Q_{XY}|y^n)e^{-n\ell(Q_{XY})},
\label{eq:S}
\end{equation}
and so, $P_{Y^n|\calC}(y^n)=S(y^n)/M$.
For $y^n$ of type $\hat{P}_{y^n}=Q_Y$,
using $P_Y^{\otimes n}(y^n)= e^{-n[H_Q(Y)+D_{\mbox{\tiny m}}(Q_Y)]}$:
\begin{equation}
\Lambda(y^n)=
\frac{1}{n}\log S(y^n) - R + H_Q(Y) + D_{\mbox{\tiny m}}(Q_Y).
\label{eq:LLR_S}
\end{equation}

To analyze the two kinds of probability of error, we shall invoke results
concerning the large deviations behavior of $\{N(Q_{XY})\}$. Since these are
binomial random variables, the following theorems from \cite{MW25} (with a
slight change in notation) will be useful.

\begin{enumerate}[label=(T\arabic*),leftmargin=2.5em,noitemsep]
\item\label{T1} \textbf{Theorem~4.1 of \cite{MW25}}.
For $N\sim\mathrm{Binomial}(e^{nA},e^{-nB})$ ($A>0$ and $B>0$) and $C\in\mathbb{R}$:
\begin{align}
\lim_{n\to\infty}-\tfrac{1}{n}\ln\Pr\{N>e^{nC}\}
    &= \begin{cases}[B-A]_+ & [A-B]_+\ge C,\\
         \infty & \text{elsewhere},\end{cases}
    \label{eq:T1u}\\
\lim_{n\to\infty}-\tfrac{1}{n}\ln\Pr\{N<e^{nC}\}
    &= \begin{cases}0 & A-B<C,\\
         \infty & A-B>C.\end{cases}
    \label{eq:T1l}
  \end{align}

When $A>B$, $\Pr\{N>e^{nC}\}\doteq 1$ for
$C\le A-B$, and $\Pr\{N<e^{nC}\}$ is doubly exponentially
small for $C<A-B$. When $A<B$ and $C<0$, $\Pr\{N>e^{nC}\}\doteq e^{-n(B-A)}$.
We note in passing that the case $A=B$ is the case of asymptotic Poissonianity.
In particular, if $N\sim\mathrm{Binomial}(e^{nA},\mu e^{-nA})$, then
as $n\to\infty$, the distribution of $N$ tends to a Poissonian random variable with parameter
$\mu$, i.e., $\Pr\{N=k\}\to\frac{\mu^ke^{-\mu}}{k!}$, $k=0,1,2,\ldots$.

\item\label{T3} \textbf{Theorem~4.3 of \cite{MW25}}.
Let $N_j\sim\mathrm{Binomial}(e^{nA_j},e^{-nB_j})$, $j=1,\ldots,k_n$,
with $k_n\doteq 1$.
Then for any thresholds $C_j\in\mathbb{R}$:
\begin{equation}
\Pr\Bigl\{\bigcap_j\{N_j\le e^{nC_j}\}\Bigr\}
    \;\doteq\;
    \mathbf{1}\Bigl\{\min_j(B_j-A_j+[C_j]_+)>0\Bigr\}.
    \label{eq:T3}
  \end{equation}
When $\min_j (B_j-A_j+[C_j]_+)>0$, every $\Pr\{N_j>e^{nC_j}\}$ is either
exponentially or doubly
exponentially small (by Theorem~4.1 of \cite{MW25}), and so, the union bound on
the complementary events yields
an intersection probability tending to $1$. On the other hand,
when $\min_j (B_j-A_j+[C_j]_+)\le 0$, then at least for some $j$,
$B_j-A_j+[C_j]_+\le 0$, and for that $j$,
$\Pr\{N_j\le e^{nC_j}\}\doteq 0$, which implies that
the intersection probability is a-fortiori $\doteq 0$.
Note that there are no assumptions on the statistical dependence or
independence among $\{N_j\}$.
\end{enumerate}

It should be noted that in \cite[Theorem 4.3]{MW25}, the assertion is
formulated with a common threshold parameter $C$, rather than various
thresholds $C_j$, as stated here. Nevertheless, the proof in \cite{MW25}
extends straightforwardly to allow different thresholds.

In view of those results, it is apparent that given $y^n$, a type $Q_{XY}$
with $\IQ< R$ is a type that is typically populated with exponentially many
codewords, in particular, $N(Q_{XY}|y^n)$ concentrates at the exponential
order of $e^{n[R-\IQ]}$. Such a type will be henceforth referred to as a {\em
bulk type}. By contrast, a type for which $\IQ>R$, is rarely populated, as the
probability of for $N(Q_{XY}|y^n)\ge 1$ is of the exponential order of 
$e^{-n[\IQ-R]}\to 0$. Such a type will be referred to as a {\em sparse type}.

\section{Main Result}
\label{sec:main}

For a joint type $Q=Q_{XY}$ with $Q_X=P_X$ and rate $R\ge 0$, let us
define
\begin{equation}
\lambda(Q_{XY},R)= D_{\mbox{\tiny m}}(Q_Y) - D_{\mbox{\tiny c}}(Q_{XY}) + [\IQ-R]_+.
\label{eq:lambda}
\end{equation}
By the DPI (\ref{dataprocessing}),
$\Dm\le \Dc$.
Hence for bulk types ($\IQ\le R$): $\lambda(Q_{XY},R)=D_{\mbox{\tiny
m}}(Q_Y)-D_{\mbox{\tiny c}}(Q_{XY})\le 0$.
For the true channel $Q_{Y|X}=W$, $D_{\mbox{\tiny c}}(P_X\otimes W)=D(W\|W|P_X)=0$ and
$D_{\mbox{\tiny m}}(Q_Y)=D(P_Y\|P_Y)=0$, and so,
$\lambda(P_X\otimes W,R)=[I(X;Y)-R]_+$,
which is positive for $R<I(X;Y)$, and zero for $R\ge I(X;Y)$.

Our main theorem, whose proof appears in the appendix, is the following.

\begin{theorem}\label{thm:main}
For all $R\ge 0$ and $\tau\in\mathbb{R}$:
\begin{equation}
\EFA(\tau,R)
=\min_{\substack{Q:~Q_X=P_X\\\lambda(Q_{XY},R)\ge\tau}}
\bigl\{D_{\mbox{\tiny m}}(Q_Y)+[\IQ-R]_+\bigr\}.
\label{eq:FA_exp}
\end{equation}
For $R>0$ and $\tau\in\mathbb{R}$:
\begin{equation}
\EMD(\tau,R)
=\min_{\substack{Q_{XY}:\,Q_X=P_X\\\lambda(Q_{XY},R)<\tau\\\Delta(Q_Y,R)<\tau}}\Dc,
\label{eq:MD_exp}
\end{equation}
where
\begin{equation}
\Delta(Q_Y,R)=\max_{\substack{Q'_{XY}:\,Q'_X=P_X,\,Q'_Y=Q_Y\\I_{Q'}(X;Y)\le R}}
\bigl[\Dm-D_{\mbox{\tiny c}}(Q'_{XY})\bigr].
\end{equation}
For $\tau>0$, the constraint $\Delta(Q_Y,R)<\tau$ is redundant (due to the DPI
(\ref{dataprocessing}))
and the remaining active constraints are $Q_X=P_X$ and $\lambda(Q_{XY},R)<\tau$.
\end{theorem}

Two remarks concerning this theorem are in order.

\begin{remark}
\label{rem:jump}
For $\tau<0$, missed detection requires the transmitted codeword
to have a rare joint type, which is an unlikely event whose probability
decays exponentially according to $\Dc$.
Given such an atypical transmitted pair, the interfering codewords
are too weak to push the likelihood ratio above the (negative) threshold,
and so, missed detection occurs with probability $\doteq 1$ conditional on
this rare type.
Hence $\EMD(\tau,R)$ is finite and positive for $\tau<0$ (when $R<I(X;Y)$),
governed by the same formula~\eqref{eq:MD_exp} as for $\tau>0$.
As $\tau\to-\infty$: the feasible set eventually empties
and $\EMD(\tau,R)\to+\infty$.
\end{remark}

\begin{remark}[$R=0$: one codeword]\label{rem:r0}
For $R=0$ ($M=1$), the problem reduces to simple hypothesis testing
of $W^n(\cdot|x^n(1))$ vs.\ $P_Y^{\otimes n}$ averaged over
$x^n(1)\sim\mathrm{Uniform}(\calT(P_X))$.
With a single codeword, $S(y^n)=e^{-n\ell(Q_{XY})}$ is deterministic,
and both formulas~\eqref{eq:FA_exp}--\eqref{eq:MD_exp} hold with $R=0$
(the $\Delta$ constraint is vacuous since there are no interfering codewords).
In particular, $\EMD(\tau,0)$ is finite for all $\tau\in\mathbb{R}$
with no discontinuity, in contrast to $R>0$.
\end{remark}

\section{Properties and Phase Structure}
\label{sec:phase}

We begin with an elementary observation that follows directly
from Theorem~\ref{thm:main}.

\begin{corollary}[Soft covering]\label{prop:softcov}
For all $R\ge I(X;Y)$: $\EFA(0,R)=\EMD(0,R)=0$.
\end{corollary}

To see why this is true, observe that
the type $Q_{XY}=P_X\otimes W$ has $\Dm=0$,
$I_Q(X;Y)=I(X;Y)\le R$, $\Dc=0$, and $\lambda(Q_{XY},R)=0$,
so it is feasible for $\tau=0$ with zero cost, hence $\EFA(0,R)=0$. Likewise,
$\EMD(0,R)=0$ since types arbitrarily close to $P_X\otimes W$ with $\lambda(Q_{XY},R)<0$
and $\Dc\to 0$ exist by continuity.

This is the Neyman--Pearson implication of the soft covering
lemma~\cite{HanVerdu93}: at $\tau=0$, both error
exponents simultaneously vanish when $R\ge I(X;Y)$,
meaning the two hypotheses are exponentially indistinguishable.
In other words, even if the FA and MD probabilities decay as $n\to\infty$,
in this case, the rates of their decay are definitely slower than exponential.

\subsection{Phase transitions}

Both exponents have several regions of behavior as functions of $\tau$.
The following two propositions describe them.

\begin{proposition}[FA: flat and active regions]\label{prop:flat}
Let $Q^*_{XY}$ minimize $\Dm+[I_Q(X;Y)-R]_+$
over all $Q_{XY}$ with $Q_X=P_X$, and
define
\begin{equation}
  \tau_{\rm flat}(R):=\lambda(Q^*_{XY},R),
  \qquad
  \lambda_{\max}(R):=\max_{Q_{XY}:\,Q_X=P_X}\lambda(Q_{XY},R).
\end{equation}
$\EFA(\tau,R)$ has three regions:
\begin{enumerate}[label=(\roman*),noitemsep]
\item \emph{Flat:} $\tau\le\tau_{\rm flat}(R)$:
  $\EFA(\tau,R)=D_{\mbox{\tiny m}}(Q^*_{XY})+[I_{Q^*}(X;Y)-R]_+$ (constant in $\tau$).
\item \emph{Active:} $\tau\in(\tau_{\rm flat}(R),\lambda_{\max}(R))$:
  $\EFA(\tau,R)$ strictly increases in $\tau$.
\item \emph{Infinite:} $\tau>\lambda_{\max}(R)$: $\EFA(\tau,R)=+\infty$.
\end{enumerate}
\end{proposition}

\begin{proof}
(i) For $\tau\le\tau_{\rm flat}(R)=\lambda(Q^*_{XY},R)$,
the unconstrained minimizer $Q^*_{XY}$ satisfies $\lambda(Q^*_{XY},R)\ge\tau$
and hence is feasible. Since it achieves the global minimum, the
constrained minimum equals the unconstrained one; $\EFA(\tau,R)$ is flat.
(ii) For $\tau>\tau_{\rm flat}(R)$, $Q^*_{XY}$ is infeasible.
The constrained minimum strictly exceeds the unconstrained one and
strictly increases with $\tau$ as the feasible set shrinks.
(iii) For $\tau>\lambda_{\max}(R)$, every type has
$\lambda(Q_{XY},R)\le\lambda_{\max}(R)<\tau$, so the feasible set is
empty and $\EFA(\tau,R)=+\infty$ by convention.
\end{proof}

\begin{proposition}[MD: zero, active, and divergent regions]\label{prop:MD_regions}
Let $\tau^*(R)=[I(X;Y)-R]_+$ and
$\lambda_{\min}(R)=\min_{Q_{XY}:\,Q_X=P_X}\lambda(Q_{XY},R)$.
For $R>0$, $\EMD(\tau,R)$ has three regions:
\begin{enumerate}[label=(\roman*),noitemsep]
\item \emph{Zero:} $\tau\ge\tau^*(R)$: $\EMD(\tau,R)=0$.
\item \emph{Active:} $\tau\in(\lambda_{\min}(R),\tau^*(R))$:
  $\EMD(\tau,R)$ is finite, positive, and strictly decreasing in $\tau$.
\item \emph{Infinite:} $\tau\le\lambda_{\min}(R)$: $\EMD(\tau,R)=+\infty$
  (the feasible set $\{\lambda(Q_{XY},R)<\tau\}$ is empty).
\end{enumerate}
Within the active region, there is a kink at some
$\tau_{\rm kink}(R)\in(\lambda_{\min}(R),0)$ where the minimizing type
transitions from bulk ($I_{Q^*}(X;Y)\le R$, for $\tau<\tau_{\rm kink}$)
to sparse ($I_{Q^*}(X;Y)>R$, for $\tau>\tau_{\rm kink}$).
\end{proposition}

\begin{proof}
(i) For $\tau\ge\tau^*(R)$, types near $P_X\otimes W$
(with $\Dc\to 0$ and $\lambda(Q_{XY},R)\to\tau^*(R)^-$) are feasible,
giving $\inf\Dc=0$ and $\EMD(\tau,R)=0$.
ii) Monotonicity in $\tau$ is immediate: the feasible set
$\{\lambda(Q_{XY},R)<\tau\}$ grows as $\tau$ increases, so $\EMD(\tau,R)$ is non-increasing.
Strict decrease and positivity follow from the formula~\eqref{eq:MD_exp}.
(iii) For $\tau\le\lambda_{\min}(R)$, every type has
$\lambda(Q_{XY},R)\ge\lambda_{\min}(R)\ge\tau$, so the feasible set is
empty and $\EMD(\tau,R)=+\infty$.
\end{proof}

\subsection{The Neyman--Pearson tradeoff zone}

The results above combine into a clean picture of the Neyman--Pearson tradeoff
(Corollary~\ref{cor:tradeoff} and Proposition~\ref{prop:MD_regions}).

\begin{corollary}\label{cor:tradeoff}
Let $\tau^*(R):=[I(X;Y)-R]_+$.
\begin{enumerate}[label=(\roman*),noitemsep]
\item For $R<I(X;Y)$: $\EMD(\tau,R)>0$ for all $\tau<\tau^*(R)$,
  and $\EMD(\tau,R)=0$ for all $\tau\ge\tau^*(R)$.
  Both $\EFA(\tau,R)>0$ and $\EMD(\tau,R)>0$ simultaneously
  for all $\tau\in(0,\tau^*(R))$. 
\item For $R\ge I(X;Y)$: $\tau^*(R)=0$, so $\EMD(\tau,R)=0$
  for all $\tau\ge 0$ and $\EFA(\tau,R)=0$ for all $\tau\le 0$.
  Both exponents simultaneously vanish at $\tau=0$
  (Corollary~\ref{prop:softcov}), and the Neyman--Pearson tradeoff zone is empty.
\end{enumerate}
\end{corollary}

\begin{proof}
\emph{$\EMD(\tau,R)=0$ iff $\tau\ge\tau^*(R)$.}
The only type with $\Dc=0$ is $Q_{XY}=P_X\otimes W$,
which has $\lambda(P_X\otimes W,R)=\tau^*(R)$.
For $\tau\ge\tau^*(R)$, types approaching $P_X\otimes W$ have
$\Dc\to 0$ and $\lambda\to\tau^*(R)^-<\tau$, so they are feasible
with $\Dc\to 0$, giving $\EMD(\tau,R)=0$.
For $\tau<\tau^*(R)$, every type with small $\Dc$ has $Q_{Y|X}$ close to $W$
and hence $\lambda(Q_{XY},R)$ close to $\tau^*(R)>\tau$, so it is not feasible;
thus $\inf_{\lambda(Q_{XY},R)<\tau}\Dc>0$ and $\EMD(\tau,R)>0$.

\emph{$\EFA(\tau,R)>0$ for all $\tau>0$.}
A zero-cost type for $\EFA(\tau,R)$ requires $Q_Y=P_Y$ and $I_Q(X;Y)\le R$.
Any such type has $\lambda(Q_{XY},R)=-\Dc\le 0$,
so it is feasible for $\EFA(\tau,R)$ (i.e.\ $\lambda(Q_{XY},R)\ge\tau$) only if $\tau\le 0$.
Hence for $\tau>0$ no zero-cost type is feasible and $\EFA(\tau,R)>0$.

\emph{Conclusion.}
For $R<I(X;Y)$: both $\EFA(\tau,R)>0$ (since $\tau>0$) and $\EMD(\tau,R)>0$
(since $\tau<\tau^*(R)$) hold simultaneously iff $\tau\in(0,\tau^*(R))$.
For $R\ge I(X;Y)$: $\tau^*(R)=0$, so $\EMD(\tau,R)=0$ for all $\tau\ge 0$
and $\EFA(\tau,R)=0$ for all $\tau\le 0$ (since $P_X\otimes W$ has
$\Dc=0$, $Q_Y=P_Y$, $\IQ=I(X;Y)\le R$, giving a feasible zero-cost
type for $\EFA(\tau,R)$ whenever $\tau\le 0$).
Thus the two zero regions cover all of $\mathbb{R}$ and the tradeoff zone
is empty.
\end{proof}

\section{Numerical Illustration: Z-Channel}
\label{sec:numerics}

In the section, we demonstrate the behavior of both exponents
for a numerical example of a Z-channel.

Let $\calX=\calY=\{0,1\}$ and consider the Z-channel with the following
input-output transition probabilities:
\begin{eqnarray}
W(0|0)&=&1,\\
W(1|0)&=&0,\\
W(0|1)&=&w,~~~\mbox{and}\\
W(1|1)&=&1-w,
\end{eqnarray}
with $w=0.45$ and $P_X(0)=P_X(1)=0.5$. The corresponding mutual information is
$I(X;Y)=0.2441$ nats per channel use.
Joint types are parameterized by $q=Q(0|1)\in(0,1)$, with
$Q(y|0)=W(y|0)$ forced. Defining the binary entropy function and the binary KL
divergence as
\begin{eqnarray}
h_{\rm b}(u)&=&-u\ln u-(1-u)\ln(1-u)\\
D_{\rm b}(u\|v))&=&u\ln \frac{u}{v}+(1-u)\ln\frac{1-u}{1-v},
\end{eqnarray}
where $(u,v)\in[0,1]^2$, this gives
\begin{eqnarray}
I_Q(X;Y)&=&h_{\rm b}
\left(\tfrac{1+q}{2}\right)-\tfrac{1}{2}h_{\rm b}(q),\\
\Dc&=&\tfrac{1}{2}D_{\rm b}(q\|w),~~~~\mbox{and}\\
\Dm&=&D_{\rm b}\left(\tfrac{1+q}{2}\big\|\tfrac{1+w}{2}\right).
\end{eqnarray}
Seven figures follow, one per page, and each one contains a brief description
and discussion, in addition to the figure caption.

\clearpage
\begin{figure}[t]
  \centering
  \includegraphics[width=0.85\textwidth]{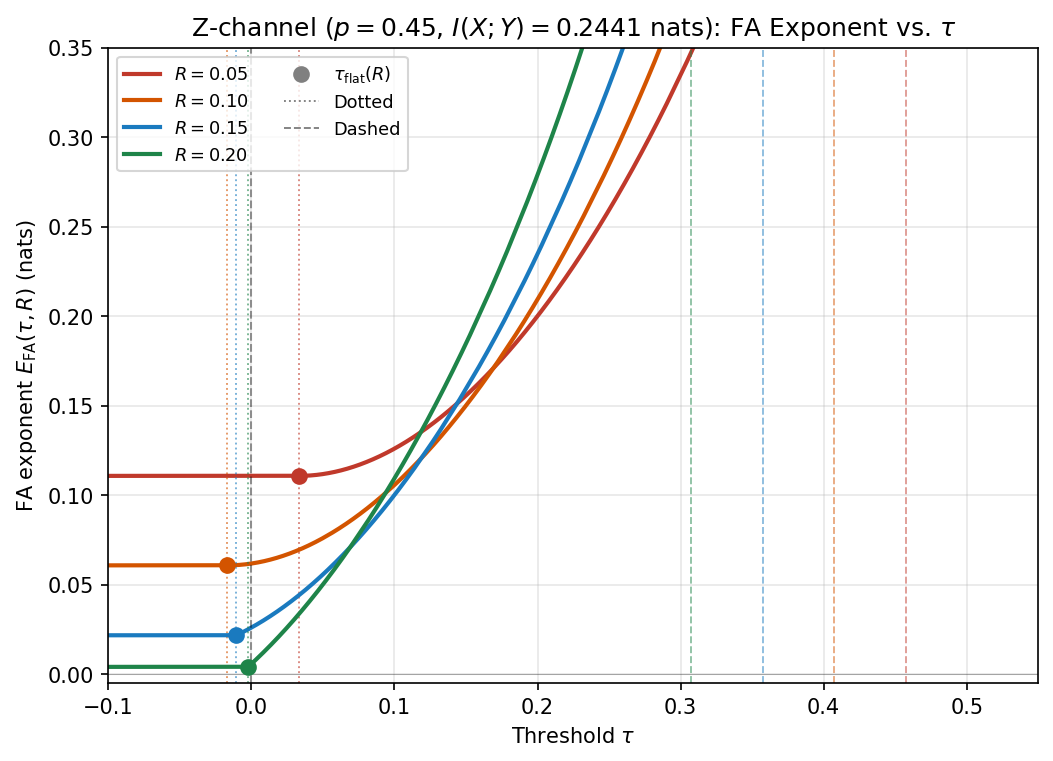}
  \caption{\textbf{FA exponent $\EFA(\tau,R)$} vs.\ $\tau$ for four rates.}
  \label{fig:FA}
\end{figure}

In Fig.\ \ref{fig:FA},
for each rate $R\in\{0.05,0.10,0.15,0.20\}$ (all below $I(X;Y)=0.2441$),
$\EFA(\tau,R)$ is flat at $E_{\mathrm{FA,flat}}(R)=D_{\mbox{\tiny
m}}(Q^*)+[I_{Q^*}(X;Y)-R]_+>0$
for $\tau\le\tau_{\rm flat}(R)$ (filled dots, dotted verticals),
then strictly increases for $\tau>\tau_{\rm flat}(R)$,
diverging to $+\infty$ at $\lambda_{\max}(R)$ (dashed verticals).
$\EFA(\tau,R)$ is continuous throughout.
The flat value $E_{\mathrm{FA,flat}}(R)>0$ (rather than $0$) is unique to
singular channels like the Z-channel with $W(1|0)=0$: a typical output $y^n$
(which has many 1's) can only be explained by codewords that have `1'
at every position where $y^n$ has `1', and with only $e^{nR}$ codewords
drawn i.i.d.\ from $P_X$, the probability that even a single such
compatible codeword exists is exponentially small.
Hence $P_{Y^n|\calC}(y^n)=0$ with high probability over the ensemble of codes,
$\Lambda(y^n)=-\infty$, and FA is exponentially rare
for \emph{any} finite~$\tau$, even $\tau\to-\infty$.
See Figure~\ref{fig:FA_zoom} for a zoom on the active region.

\clearpage
\begin{figure}[t]
  \centering
  \includegraphics[width=0.85\textwidth]{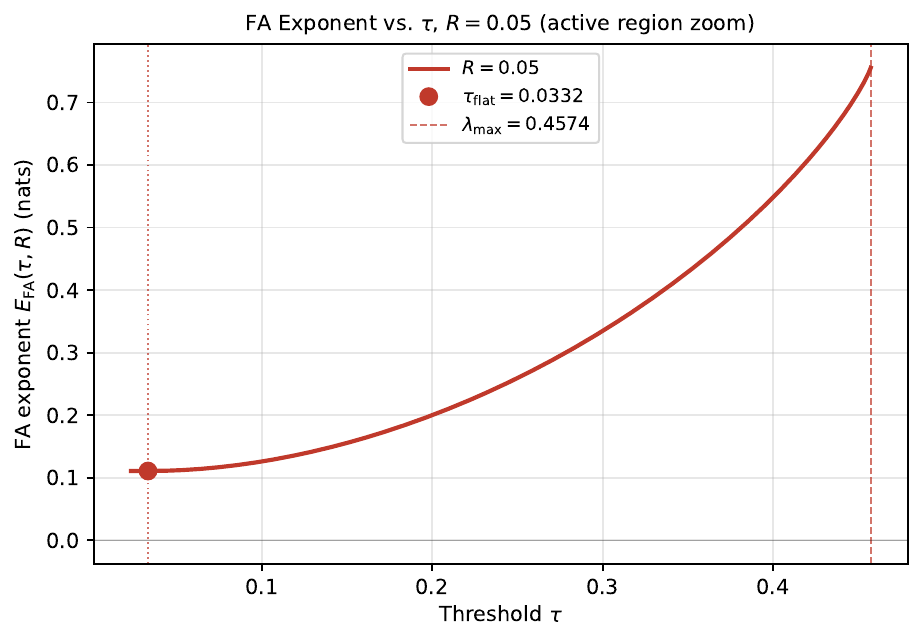}
  \caption{\textbf{Zoom into the active region of $\EFA(\tau,R)$}.
  Normalized axes: $x=0$ at $\tau_{\rm flat}(R)$, $x=1$ at
  $\lambda_{\max}(R)$. Only $R=0.05$ (width $\approx 0.032$ nats) is visible.}
  \label{fig:FA_zoom}
\end{figure}

Fig.\ \ref{fig:FA_zoom} provides a zoom into 
the active region $\tau\in(\tau_{\rm flat}(R),\lambda_{\max}(R))$
for $R=0.05$.
At $\tau_{\rm flat}(R)\approx 0.033$ (filled dot): $\EFA(\tau,R)$ lifts off
its flat value $\approx 0.111$ and begins to increase.
At $\lambda_{\max}(R)\approx 0.457$ (dashed vertical): $\EFA(\tau,R)\to+\infty$.
$\EFA(\tau,R)$ is strictly convex and monotone increasing throughout the active region.

\clearpage
\begin{figure}[t]
  \centering
  \includegraphics[width=0.85\textwidth]{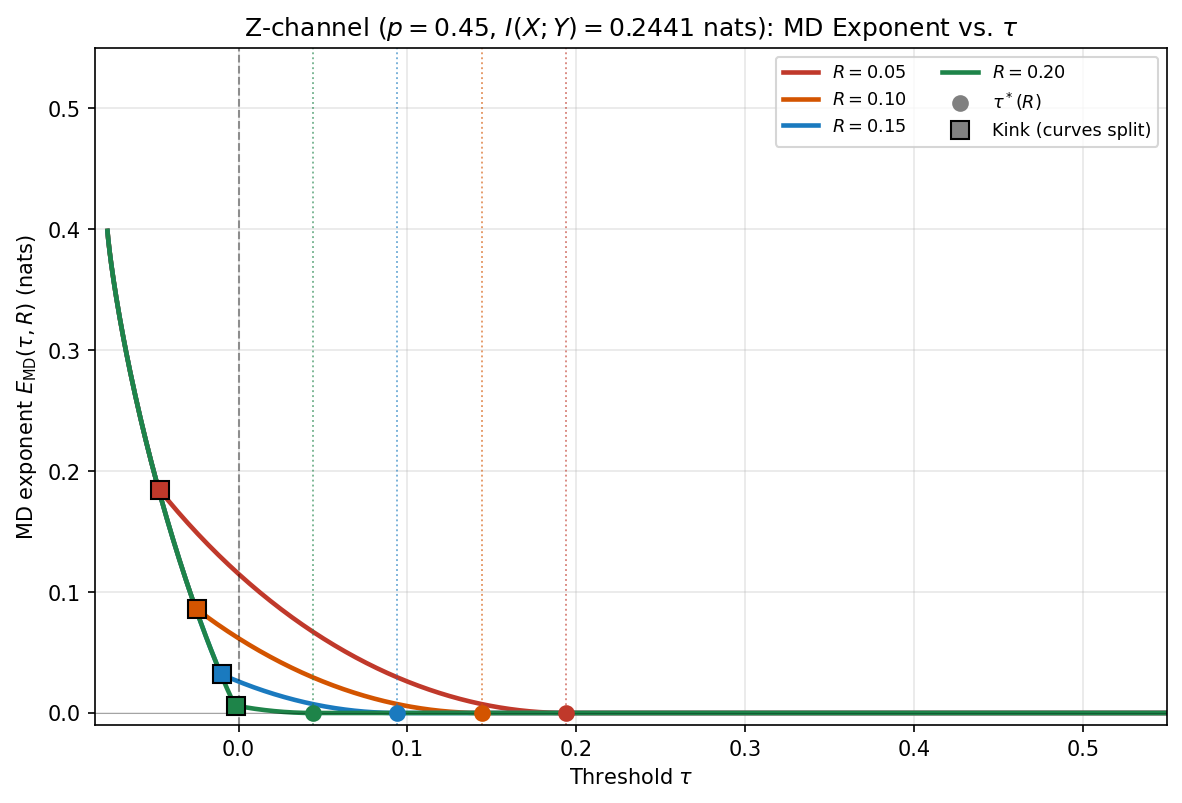}
  \caption{\textbf{MD exponent $\EMD(\tau,R)$} vs.\ $\tau$
  for four rates $R\in\{0.05,0.10,0.15,0.20\}$.
  Filled circles: $\tau^*(R)=[I(X;Y)-R]_+$ (onset of $\EMD(\tau,R)=0$).}
  \label{fig:MD}
\end{figure}

In Fig.\ \ref{fig:MD}
we see the MD analogue of Fig.\ \ref{fig:FA}.
$\EMD(\tau,R)$ is monotone non-increasing in $\tau$ for each 
one of four rates as before.
The curves start from $\tau\approx\lambda_{\min}\approx -0.078$ nats
(minimum achievable $\lambda(Q_{XY},R)$; $\EMD(\tau,R)=+\infty$ for smaller $\tau$).

\emph{Common branch.}
For $\tau$ sufficiently negative (below the leftmost kink at
$\tau_{\rm kink}(R=0.05)\approx-0.047$), all four curves are
\emph{exactly identical}: the minimizing type $Q^*(\tau)$ has
$I_{Q^*}(X;Y)<R$ for all four rates, so the rate constraint is
inactive and $\lambda(Q_{XY}^*,R)=D_{\mbox{\tiny m}}(Q^*_Y)-D_{\mbox{\tiny c}}(Q^*_{XY})$ is the same
for all $R$.

\emph{Sequential splitting.}
As $\tau$ increases, the curves split one by one at the kink points
$\tau_{\rm kink}(R)$ (filled squares), where $I_{Q^*}(X;Y)=R$ exactly
and the minimizing type transitions from bulk to sparse.
This is a first-order phase transition: the slope changes
discontinuously, with the sparse branch (higher rate of decrease)
taking over above the kink.

For $0<\tau<\tau^*(R)=[I(X;Y)-R]_+$: $\EMD(\tau,R)>0$.
At $\tau=\tau^*(R)$ (filled circles, dotted verticals): $\EMD(\tau,R)=0$.

\clearpage
\begin{figure}[t]
  \centering
  \includegraphics[width=0.85\textwidth]{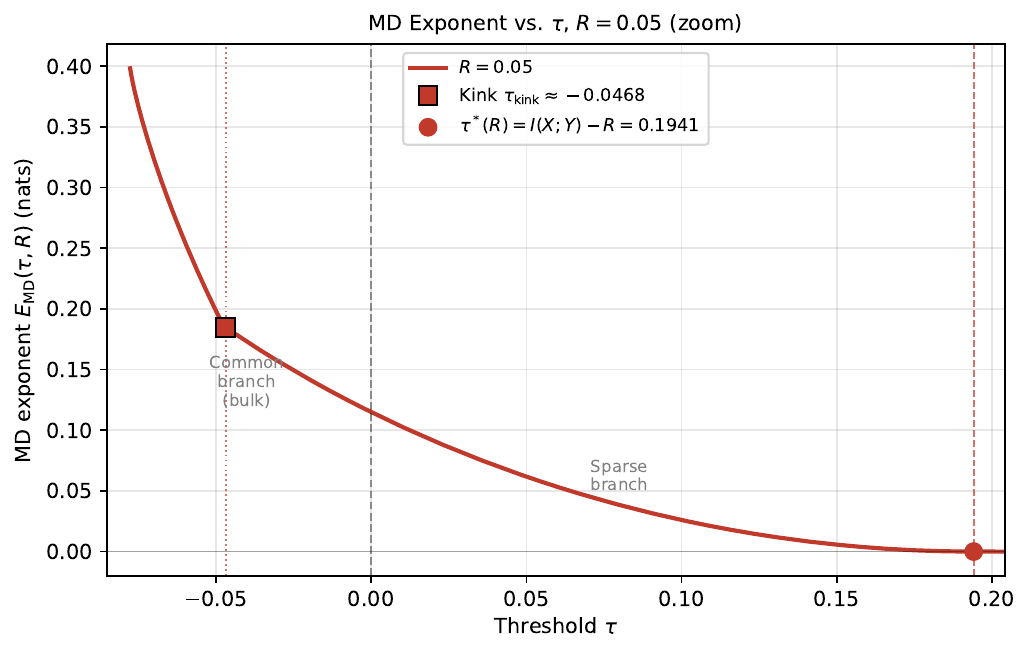}
  \caption{\textbf{Zoom into the active region of $\EMD(\tau,R)$, $R=0.05$.}
  The kink at $\tau_{\rm kink}(R)\approx -0.047$ (filled square, dotted vertical)
  marks the transition from the common bulk branch (left)
  to the sparse branch (right, rate-dependent).
  Filled circle: $\tau^*(R)=I(X;Y)-R\approx 0.194$ (onset of $\EMD(\tau,R)=0$),
  dashed vertical.}
  \label{fig:MD_zoom}
\end{figure}

Fig.\ \ref{fig:MD_zoom}
provides a zoom-in picture on $\EMD(\tau,R)$ for $R=0.05$, revealing the full structure
of the curve from $\lambda_{\min}(R)$ to $\tau^*(R)$.
Left of the kink $\tau_{\rm kink}\approx-0.047$: minimizing type is bulk,
common to all rates.
Right of the kink: minimizer becomes sparse, curves split by rate.
$\EMD(\tau,R)$ is monotone non-increasing throughout, reaching zero at $\tau^*(R)$.

\clearpage
\begin{figure}[t]
  \centering
  \includegraphics[width=0.9\textwidth]{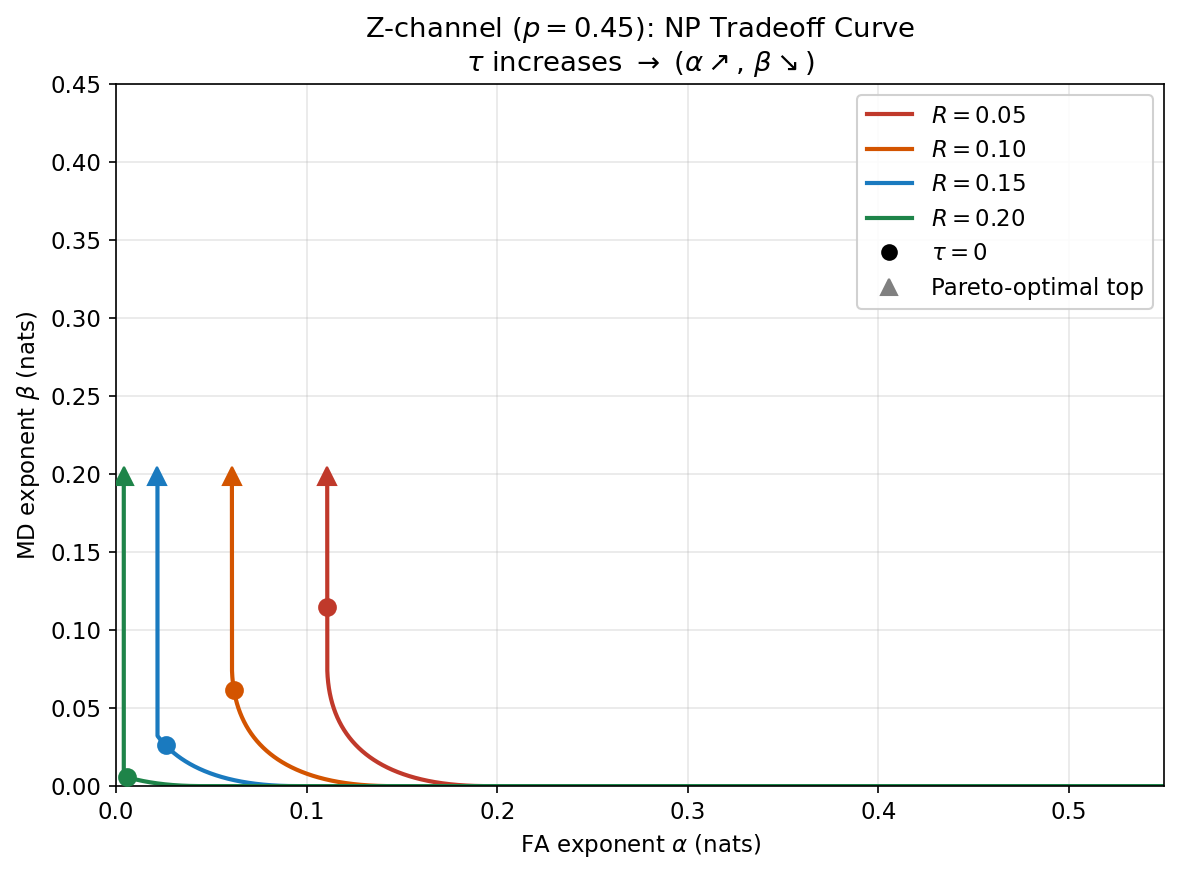}
  \caption{\textbf{Neyman--Pearson tradeoff curve}: $\EMD(\tau,R)$ vs.
$\EFA(\tau,R)$,
  parametrized by $\tau$ ($\tau$ increasing: $\EFA(\tau,R)\nearrow$,
$\EMD(\tau,R)\searrow$).
  \textbf{Left}: raw parametric curve; vertical segments arise because
  $\EFA(\tau,R)$ is flat while $\EMD(\tau,R)$ decreases. Triangles: top of each vertical.
  \textbf{Right}: upper-envelope curve (each flat segment collapsed to
  its highest point).}
  \label{fig:tradeoff}
\end{figure}

Fig.\ \ref{fig:tradeoff} gives 
the parametric curve $(\EFA(\tau,R),\EMD(\tau,R))$ as $\tau$ sweeps the
tradeoff range. It is shown for four rates.
Each curve has a vertical segment where $\EFA(\tau,R)$ is flat
(Proposition~\ref{prop:flat}) while $\EMD(\tau,R)$ decreases: during this portion
there is no genuine tradeoff, only a loss in MD performance at fixed $\EFA(\tau,R)$.
The right panel collapses each vertical segment to its highest point,
revealing the genuine Neyman--Pearson tradeoff: to achieve larger $\EFA(\tau,R)$
one must sacrifice $\EMD(\tau,R)$.

\clearpage
\begin{figure}[t]
  \centering
  \includegraphics[width=0.9\textwidth]{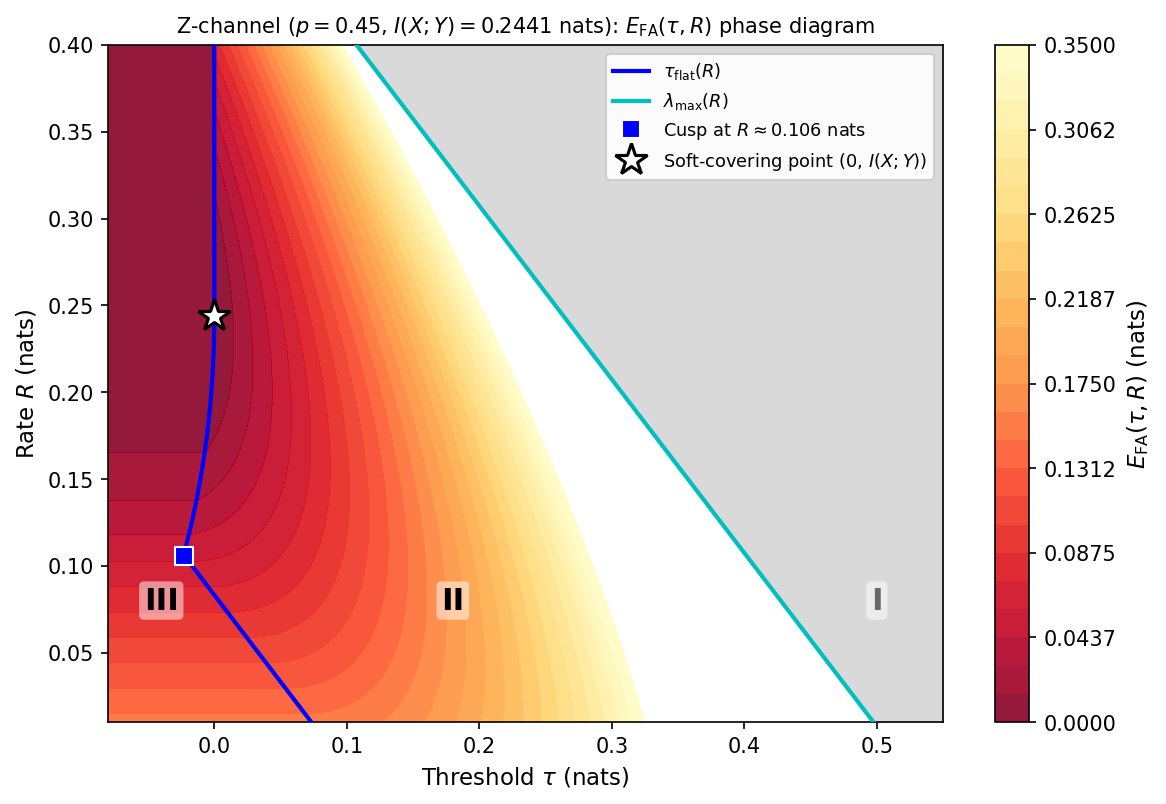}
  \caption{\textbf{FA exponent $\EFA(\tau,R)$}: phase diagram
  in the $(\tau,R)$ plane (Proposition~\ref{prop:flat}).
  Region~\textbf{III} (left of blue curve, $\tau\le\tau_{\rm flat}(R)$):
  $\EFA$ is flat in $\tau$ (for fixed $R$) but varies with $R$.
  Region~\textbf{II} (between blue and cyan curves,
  $\tau_{\rm flat}(R)<\tau<\lambda_{\max}(R)$): $\EFA$ strictly increasing.
  Region~\textbf{I} (grey, right of cyan curve, $\tau>\lambda_{\max}(R)$):
  $\EFA=+\infty$.
  Blue curve: $\tau_{\rm flat}(R)$; cyan curve: $\lambda_{\max}(R)$.
  Blue square: cusp in $\tau_{\rm flat}(R)$ at $R\approx 0.106$ nats.
  White star: soft-covering point $(0,\,I(X;Y))$,
  at the boundary between Regions~II and~III.}
  \label{fig:alpha_phase}
\end{figure}

Fig.\ \ref{fig:alpha_phase} displays the three regions of
Proposition~\ref{prop:flat} in the $(\tau,R)$ plane.
Moving from left to right in $\tau$: Region~III (left of the blue
curve $\tau_{\rm flat}(R)$, where $\EFA$ is flat in $\tau$ for
fixed $R$ but varies with $R$) transitions into
Region~II (strictly increasing $\EFA$, between the blue and
cyan curves), and finally into Region~I ($\EFA=+\infty$,
grey, to the right of the cyan curve $\lambda_{\max}(R)$).
The white star marks the soft-covering point $(\tau,R)=(0,I(X;Y))$;
it lies at the boundary between Regions~II and~III
(since $\tau_{\rm flat}(R)\to 0$ as $R\nearrow I(X;Y)$)
and corresponds to $\EFA=E_{\mathrm{FA,flat}}(R)>0$, not zero.
The blue square marks the cusp in $\tau_{\rm flat}(R)$ at
$R=R_{\rm cr}\approx 0.106$ nats, where the unconstrained minimizer
$Q^*_{XY}$ transitions from a sparse type ($I_{Q^*}>R$, slope $-1$)
to a bulk type, creating the visible kink in the blue curve.

\clearpage
\begin{figure}[t]
  \centering
  \includegraphics[width=0.9\textwidth]{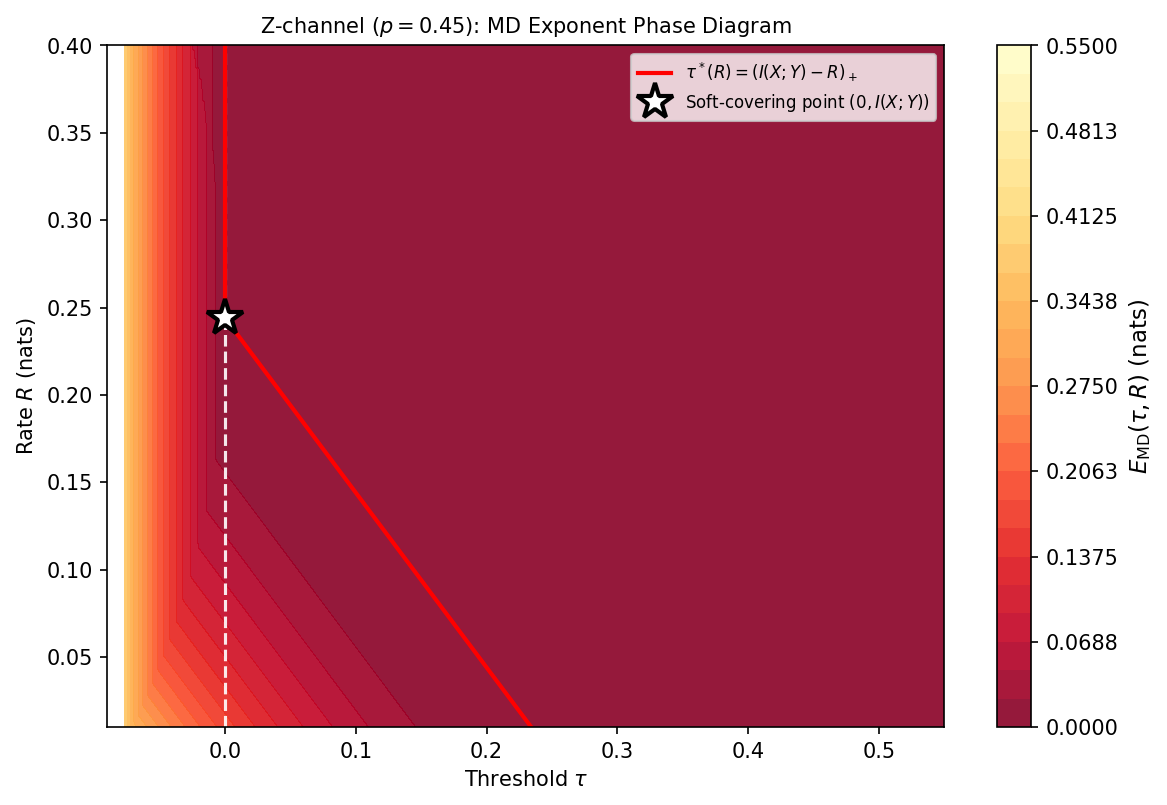}
  \caption{\textbf{MD exponent $\EMD(\tau,R)$}: phase diagram.
  Shaded region ($\tau\le 0$): $\EMD(\tau,R)$ finite and positive when
  the feasible set $\{\lambda(Q_{XY},R)<\tau\}$ is non-empty
  (Remark~\ref{rem:jump}); $\EMD(\tau,R)=+\infty$ otherwise.
  Colored region ($\tau>0$): $\EMD(\tau,R)$ finite, with sharp transition
  to $0$ at $\tau^*(R)=[I(X;Y)-R]_+$ (red line).
  White dashed line: $\tau=0$. White star: soft-covering point $(0,I(X;Y))$.}
  \label{fig:beta_phase}
\end{figure}

In Fig.\ \ref{fig:beta_phase}
the shaded region ($\tau\le 0$) has finite $\EMD(\tau,R)$ given by
$\min_{\lambda(Q_{XY},R)<\tau}D_{\mbox{\tiny c}}(Q_{XY})$ when the feasible set
is non-empty, and $\EMD(\tau,R)=+\infty$ otherwise (Remark~\ref{rem:jump}).
The white dashed line at $\tau=0$ marks $\EMD(\tau,R)=0$.
The colored region ($\tau>0$) shows finite $\EMD(\tau,R)$, with a sharp
transition to $0$ at $\tau^*(R)=[I(X;Y)-R]_+$ (red line, slope $-1$).
Together, the two phase diagrams show that the Neyman--Pearson tradeoff zone
$0<\tau<\tau^*(R)$ shrinks as $R\nearrow I(X;Y)$ and vanishes at
$R=I(X;Y)$, the soft-covering threshold.

\section{Conclusion and Outlook}
\label{sec:conclusion}

We studied the soft covering lemma through the lens of Neyman--Pearson
hypothesis testing, asking: can an observer distinguish the random
codebook mixture $P_{Y^n|\calC}$ from the i.i.d.\ distribution
$P_Y^{\otimes n}$?
We derived exact single-letter formulas for the FA and MD
exponents $\EFA(\tau,R)$ and $\EMD(\tau,R)$ for all rates $R\ge 0$
and all thresholds $\tau\in\mathbb{R}$, using the type-class
enumeration toolbox of~\cite{MW25}.
The main findings are:

\begin{enumerate}[noitemsep]
\item \emph{Soft covering as a phase transition.}
At $R=I(X;Y)$ and $\tau=0$, both exponents simultaneously vanish ---
the Neyman--Pearson exponent characterization of the soft covering phenomenon.
For $R<I(X;Y)$, a genuine Neyman--Pearson tradeoff interval $(0,I(X;Y)-R)$ exists
where both exponents are positive.
For $R\ge I(X;Y)$, the two zero regions cover all of $\mathbb{R}$,
leaving no tradeoff interval.

\item \emph{Rich phase structure.}
The FA exponent $\EFA(\tau,R)$ is flat (constant in $\tau$) below
$\tau_{\rm flat}(R)$ and strictly increasing above it, diverging at
$\lambda_{\max}(R)$.
The MD exponent $\EMD(\tau,R)$ is strictly decreasing in $\tau$,
with a kink at $\tau_{\rm kink}(R)<0$ (a first-order phase transition
where the minimizing type changes from bulk to sparse),
and reaches zero at $\tau^*(R)=(I(X;Y)-R)_+$.

\item \emph{Structural asymmetry.}
The FA exponent penalizes both the empirical output deviation from
$P_Y$ and the rate surplus $[I_Q(X;Y)-R]_+$.
The MD exponent penalizes only the channel divergence $\Dc$,
with an additional constraint $\Delta(Q_Y,R)<\tau$ that identifies
which transmitted types contribute --- automatically satisfied for $\tau>0$
by the DPI.
\end{enumerate}

Several directions are open for future work.

\begin{enumerate}[noitemsep]
\item \emph{Beyond the annealed average.}
The exponents derived here are for the annealed
(codebook-averaged) probabilities $\mathbb{E}_\calC[\alpha_n]$
and $\mathbb{E}_\calC[\beta_n]$.
The quenched (almost-sure or typical-codebook) exponents
may differ and are generally harder to characterize.
The regime $R<I(X;Y)$ may exhibit a gap between annealed and quenched
exponents, as is common in disordered systems (see, e.g., \cite{Merhav26}).

\item \emph{Mismatched and compound channels.}
The detector here is the optimal Neyman--Pearson test based on the true channel $W$.
Characterizing the exponents under a mismatched detector
(one that assumes a wrong channel $\tilde W$) would be natural,
especially in the context of covert communications
where the adversary may not know the exact codebook distribution.

\item \emph{Finite-block-length refinements.}
The exponent results give the leading exponential term.
Sub-exponential (polynomial) pre-factors, as in~\cite{MW25},
would sharpen the approximation and may reveal further phase transitions
at the second order.

\item \emph{Other channels and source models.}
The analysis here is for discrete memoryless channels.
Extensions to continuous alphabet channels (most notably, Gaussian), channels with memory,
or multi-terminal settings (broadcast, multiple access) would
broaden the scope and may uncover new structural phenomena.
\end{enumerate}

\section*{Appendix -- Proof of Theorem~\ref{thm:main}}
\label{sec:proof}
\renewcommand{\theequation}{A.\arabic{equation}}
    \setcounter{equation}{0}

\subsection{The FA exponent}

We prove~\eqref{eq:FA_exp} for $R>0$;
the case $R=0$ follows by a direct type-counting argument
(see Remark~\ref{rem:r0}).

First, observe that
\begin{equation}
\alpha_n(\tau,R)
=\sum_{y^n} P_Y^{\otimes n}(y^n)\cdot
\Pr\big\{S(y^n)\ge e^{n\theta(y^n)}|y^n\bigr\},
  \label{eq:FA_start}
\end{equation}
where conditional probabilities $\Pr\{\cdot|y^n\}$ are w.r.t.\ the randomness
of $\calC$ while $y^n$ is fixed, and where,
using \eqref{eq:LLR_S}, the FA condition $\Lambda(y^n)\ge\tau$
is equivalent to $S(y^n)\ge e^{n\theta(y^n)}$ with
\begin{equation}
\theta(y^n)=\tau+R-H_Q(Y)-\Dm,
\qquad Q_Y = \hat P_{y^n}.
\label{eq:theta}
\end{equation}
Note that $\theta(y^n)$ depends on $y^n$ only through its type
$Q_Y=\hat P_{y^n}$.
Now, fix $y^n$, let $\theta=\theta(y^n)$, and $Q_Y=\hat P_{y^n}$.
Then,
\begin{equation}
\Pr\bigl\{S(y^n)\ge e^{n\theta}|y^n\bigr\}
\doteq
\Pr\left[\bigcup_{\substack{Q_{XY}:~Q_X=P_X\\Q_Y=\hat P_{y^n}}}
\bigl\{N(Q_{XY}|y^n)\ge e^{n[\theta+\ell(Q_{XY})]}\bigr\}\bigg|y^n\right],
\label{eq:FA_union}
\end{equation}
where the union is over all types, $Q_{XY}$, with marginals $Q_X=P_X$ and
$Q_Y=\hat{P}_{y^n}$. Since this is a union over a sub-exponential number of
events, this
probability is dominated by the largest probability term
of the form $\Pr\{N(Q_{XY}|y^n)\ge e^{n[\theta+\ell(Q_{XY})]}\}$. We apply
Theorem~4.1 of \cite{MW25} (see \eqref{eq:T1u}) to each
$N(Q_{XY}|y^n)\sim\mathrm{Binomial}(M,e^{-n\IQ})$, and the threshold exponent
\begin{eqnarray}
\theta+\ell(Q_{XY})
&=&\tau+R-H_Q(Y)-\Dm+H_Q(Y|X)+\Dc\nonumber\\
&=&\tau-(\Dm-\Dc)+R-\IQ.
\label{eq:threshold}
\end{eqnarray}
We next distinguish between bulk types ($\IQ\le R$) and sparse types ($\IQ>R$).
Beginning from bulk types, 
by \eqref{eq:T1u}, the upper-tail exponent is $0$ when
$R-\IQ\ge\theta+\ell(Q_{XY})$, and $\infty$ otherwise.
On substituting \eqref{eq:threshold}, we obtain
\begin{equation}
R-\IQ\ge\tau-(\Dm-\Dc)+R-\IQ.
\end{equation}
The terms $R$ and $\IQ$ on both sides cancel, giving $\Dm-\Dc\ge\tau$.
Since $\IQ\le R$, definition~\eqref{eq:lambda} gives
$\lambda(Q_{XY},R)=\Dm-\Dc$,
and so, the condition is $\lambda(Q_{XY},R)\ge\tau$.
Therefore, for bulk types,
\begin{equation}
\Pr\bigl\{N(Q_{XY}|y^n)>e^{n[\theta+\ell(Q_{XY})]}|y^n\bigr\}
\doteq
\begin{cases}1 & \lambda(Q_{XY},R)\ge\tau,\\ 
0 & \lambda(Q_{XY},R)<\tau.\end{cases}
\end{equation}
Moving on to sparse types,
eq.\ \eqref{eq:T1u} tells us that the upper-tail exponent is $\IQ-R$ when
$[R-\IQ]_+=0\ge\theta+\ell(Q_{XY})$, and $\infty$ otherwise.
Substituting \eqref{eq:threshold}, the condition $\theta+\ell(Q_{XY})\le 0$
becomes:
\begin{equation}
\tau-(\Dm-\Dc)+R-\IQ\le 0,
\end{equation}
i.e., $\Dm-\Dc+\IQ-R\ge\tau$.
Since $\IQ>R$, definition~\eqref{eq:lambda} gives
$\lambda(Q_{XY},R)=\Dm-\Dc+\IQ-R$,
so the condition is again $\lambda(Q_{XY},R)\ge\tau$.
Therefore, for sparse types,
\begin{equation}
\Pr\big\{N(Q_{XY}|y^n)>e^{n[\theta+\ell(Q_{XY})]}|y^n\bigr\}
\doteq
\begin{cases}e^{-n(\IQ-R)} & \lambda(Q_{XY},R)\ge\tau,\\ 
0 & \lambda(Q_{XY},R)<\tau.\end{cases}
\end{equation}
For each fixed marginal type $Q_Y$, we maximize over all conditional
types $Q_{X|Y}$ with $Q_X=P_X$.
The threshold $\theta=\tau+R-H_Q(Y)-\Dm$ depends on $y^n$
only through $Q_Y$ (via $\hat P_{y^n}$), while $\ell(Q_{XY})$ and
hence $\theta+\ell(Q_{XY})$ depend on the full joint type $Q_{XY}$.
After maximizing over $Q_{X|Y}$, however, the dominant exponent
depends on $y^n$ only through $Q_Y$.
We may therefore group the sum in \eqref{eq:FA_start} by types:
\begin{equation}
\alpha_n(\tau,R)
\doteq\sum_{Q_Y}
\underbrace{\sum_{y^n:~\hat P_{y^n}=Q_Y} P_Y^{\otimes
n}(y^n)}_{\doteq\,e^{-n\Dm}}
\cdot\max_{\substack{Q:\,Q_X=P_X\\Q_Y\text{ fixed}}}
\Pr\bigl\{N(Q_{XY}|y^n)\ge e^{n[\theta+\ell(Q_{XY})]}|y^n\bigr\}.
\label{eq:FA_grouped}
\end{equation}
The dominant contribution to \eqref{eq:FA_grouped} comes from the
type $Q_{XY}$ minimizing the total exponent. Since 
both cases unify as $\Dm+[I_Q(X;Y)-R]_+$ subject to $\lambda(Q_{XY},R)\ge\tau$,
this yields \eqref{eq:FA_exp}.

\subsection{The MD exponent}

Similarly as above, we prove~\eqref{eq:MD_exp} for $R>0$,
the case $R=0$ follows by a direct argument
(see Remark~\ref{rem:r0}).

Owing to the symmetry of the random coding mechanism, we assume, without loss
of generality that the transmitted codeword is $x^n(1)$.
Under $\calH_1$, the pair $(x^n(1),y^n)$ is drawn from $P_{X^n}\times W^n$.
By the method of types:
\begin{equation}
\Pr\bigl\{(x^n(1),y^n)\in\calT(Q_{XY})\bigr\}
\doteq e^{-n\Dc}
\label{eq:pair_prob}
\end{equation}
for any joint type $Q_{XY}$ with $Q_X=P_X$.
We condition on this type and treat the remaining $M-1$ codewords
$x^n(m)$, $m=2,3,\ldots,M$ as random.

We next show that if $\lambda(Q_{XY},R)\ge\tau$ then the MD event is virtually impossible.
Given that $(x^n(1),y^n)$ has joint type $Q_{XY}$ (so $y^n$ has marginal
type $Q_Y$), recall that $\theta=\tau+R-H_Q(Y)-\Dm$ as in~\eqref{eq:theta}.
Decompose
\begin{equation}
S(y^n) = e^{-n\ell(Q_{XY})} + \widetilde S(y^n),
\label{eq:S_decomp}
\end{equation}
where $e^{-n\ell(Q_{XY})}=W^n(y^n|x^n(1))$ is the deterministic contribution
of the transmitted codeword and
$\widetilde S(y^n)=\sum_{m=2}^M W^n(y^n|x^n(m))$ is the
contribution of the remaining random codewords.
The MD event is equivalent to $S(y^n)<e^{n\theta}$.
Once again, we treat bulk types and sparse types separately.

Considering sparse types first, observe that 
since $S(y^n)\ge e^{-n\ell(Q_{XY})}$, it suffices to show $-\ell(Q_{XY})\ge\theta$.
Expanding using the definitions of $\ell(Q_{XY})$ and $\theta$:
\begin{align*}
  -\ell(Q_{XY})\ge\theta
  &\iff -H_Q(Y|X)-\Dc\ge\tau+R-H_Q(Y)-\Dm\\
  &\iff \Dm-\Dc+\IQ-R\ge\tau\\
  &\iff \lambda(Q_{XY},R)\ge\tau,
\end{align*}
where the last equivalence uses the relation
$\lambda(Q_{XY},R)=\Dm-\Dc+\IQ-R$, which holds for sparse types
($\IQ>R$). This holds by assumption, so $S(y^n)\ge e^{n\theta}$
deterministically given $(x^n(1),y^n)$, and so, the MD event is impossible. 

Passing on to bulk types, 
to show that the MD event is virtually impossible when
$\lambda(Q_{XY},R)>\tau$, assume conversely that $\lambda(Q_{XY},R)>\tau$
(handling the boundary $\lambda(Q_{XY},R)=\tau$ by continuity at the end).
Consider the interfering codewords that share the \emph{same} joint type
$Q_{XY}$ with $y^n$
as the transmitted one, $x^n(1)$.
Their contribution to $\widetilde S(y^n)$ gives
\begin{equation}
S(y^n)\ge\widetilde S(y^n)\ge
\widetilde N(Q_{XY}|y^n)\cdot e^{-n\ell(Q_{XY})},
\label{eq:S_lb_bulk}
\end{equation}
where here $\widetilde N(Q_{XY}|y^n)\sim\mathrm{Binomial}(M-1,\,e^{-n\IQ})$
with $M-1\doteq e^{nR}$. Expanding $\theta+\ell(Q_{XY})$ from the definitions,
we have
\begin{align}
\theta+\ell(Q_{XY})
&=(\tau+R-H_Q(Y)-\Dm+(H_Q(Y|X)+\Dc
\notag\\
&=\tau-\bigl(\Dm-\Dc\bigr)+R-\IQ.
\label{eq:theta_plus_ell}
\end{align}
Since $\Dm-\Dc=\lambda(Q_{XY},R)\ge\tau$ and $R-\IQ\ge 0$, we get
$\theta+\ell(Q_{XY})\le R-\IQ$, with strict inequality when $\lambda(Q_{XY},R)>\tau$.
By Theorem~4.1 of \cite{MW25}~\eqref{eq:T1l}, $\Pr\{\widetilde N(Q_{XY}|y^n)
< e^{n(\theta+\ell(Q_{XY}))}|x^n(1),y^n\}$ is doubly exponentially small whenever
$\theta+\ell(Q_{XY})<R-\IQ$, i.e.\ whenever $\lambda(Q_{XY},R)>\tau$.
Hence with probability $\doteq 1$,
\begin{equation}
S(y^n)\ge\widetilde N(Q_{XY}|y^n)\cdot e^{-n\ell(Q_{XY})}
\ge e^{n[\theta+\ell(Q_{XY}])}\cdot e^{-n\ell(Q_{XY})}=e^{n\theta},
\end{equation}
so the MD event is virtually impossible in the sense that its probability
decays doubly exponentially.

In both cases, $\lambda(Q_{XY},R)\ge\tau$ implies
$\Pr\{\mathrm{MD}|x^n(1),y^n\}\doteq 0$.
Hence only types $Q_{XY}$ of $(x^n(1),y^n)$ with $\lambda(Q_{XY},R)<\tau$ can contribute
to the MD probability on the exponential scale. We next evaluate this
contribution.

Now, fix a type $Q_{XY}$ of $(x^n(1),y^n)$ with $\lambda(Q_{XY},R)<\tau$.
As mentioned earlier, the threshold $\theta=\tau+R-H_Q(Y)-\Dm$ depends on $y^n$ only through
its marginal $Q_Y$, and on $\tau$, but not on any other property of $Q_{XY}$.
For each possible joint type $Q'_{XY}$ of an interfering codeword
$x^n(m)$, $m=2,3,\ldots,M$, with $y^n$, the enumerator
$\widetilde N(Q'_{XY}|y^n)\sim\mathrm{Binomial}(M-1,e^{-nI_{Q'}(X;Y)})$.
The event $\{\widetilde S(y^n)<e^{n\theta}\}$
is equivalent to the event that $\widetilde N(Q'_{XY}|y^n)\le e^{n[\theta+\ell(Q'_{XY})]}$
for all $Q'_{XY}$ simultaneously.
To assess the probability of this intersection probability, we invoke 
Theorem~4.3 of \cite{MW25} (see \eqref{eq:T3}).
To this end, we compute $\theta+\ell(Q'_{XY})$ for each type $Q'_{XY}$ of an
interfering codeword.
Since all interfering types $Q'_{XY}$ share the marginal $Q_Y$ of $y^n$,
we have $D_{\mbox{\tiny m}}(Q_Y')=\Dm$, and so,
\begin{eqnarray}
\theta+\ell(Q'_{XY})
&=&\tau+R-H_Q(Y)-\Dm+H_{Q'}(Y|X)+D_{\mbox{\tiny c}}(Q'_{XY})\nonumber\\
&=&\tau-[D_{\mbox{\tiny m}}(Q'_Y)-D_{\mbox{\tiny c}}(Q'_{XY})]+R-I_{Q'}(X;Y).
\label{eq:theta_ell_Qp}
\end{eqnarray}
Note that this depends on $Q_{XY}'$ of the interfering codewords, and on $\tau$, but \emph{not} on 
the type $Q_{XY}$ of $(x^n(1),y^n)$ beyond its marginal $Q_Y$, which is already
absorbed into $D_{\mbox{\tiny m}}(Q'_Y)=\Dm$.
Upon applying Theorem~4.3 of \cite{MW25}, we obtain
\begin{equation}
\Pr\Bigl[
\bigcap_{Q_{XY}'}\bigl\{\widetilde N(Q'_{XY}|y^n)\le
e^{n[\theta+\ell(Q'_{XY})]}\bigg|x^n(1),y^n\bigr\}
\Bigr]
\doteq
\mathbf{1}\Bigl\{
\min_{Q_{XY}'}\bigl[I_{Q'}(X;Y)-R+[\theta+\ell(Q'_{XY})]_+\bigr]>0
\Bigr\}.
\label{eq:T3_applied}
\end{equation}
Let us define then
\begin{equation}
f(Q_{XY}')=I_{Q'}(X;Y)-R+[\theta+\ell(Q'_{XY})]_+.
\end{equation}
Using the simple identity, $a+[b-a]_+\equiv\max\{a,b\}$ with $a=I_{Q'}(X;Y)-R$ and
$b=\tau-[D_{\mbox{\tiny m}}(Q_Y)-D_{\mbox{\tiny c}}(Q'_{XY})]$, and~\eqref{eq:theta_ell_Qp},
we may rewrite the function $f$ as
\begin{equation}
f(Q'_{XY})=\max\bigl\{I_{Q'}(X;Y)-R,
\tau-[D_{\mbox{\tiny m}}(Q_Y)-D_{\mbox{\tiny c}}(Q'_{XY})]\bigr\}.
\label{eq:summand}
\end{equation}
We next simplify $f(Q'_{XY})$.
For sparse interfering types ($I_{Q'}(X;Y)>R$), 
$f(Q'_{XY})\ge I_{Q'}(X;Y)-R>0$ always.
Hence sparse types never `threaten' the occurrence of the event
$\{\min_{Q_{XY}'}f(Q_{XY}')>0\}$ and can be excluded.
For bulk interfering types, $I_{Q'}(X;Y)-R\le 0$, and so, the occurrence of the
event $\{\min_{Q_{XY}'}f(Q_{XY}')>0\}$ depends on the sign of
$\tau-[D_{\mbox{\tiny m}}(Q_Y)-D_{\mbox{\tiny c}}(Q'_{XY})]$.
Thus the condition $\min_{Q_{XY}'}f(Q_{XY}')>0$ over all $Q_{XY}'$ reduces to
\begin{equation}
\max_{\substack{Q'_{XY}:\,Q'_X=P_X,\,Q'_Y=Q_Y\\I_{Q'}(X;Y)\le R}}
\bigl[D_{\mbox{\tiny m}}(Q_Y)-D_{\mbox{\tiny c}}(Q'_{XY})\bigr]<\tau,
\end{equation}
which is exactly the condition $\Delta(Q_Y,R)<\tau$ as defined in the theorem.

For $\tau>0$, the DPI \eqref{dataprocessing} gives $D_{\mbox{\tiny c}}(Q'_{XY})\ge D_{\mbox{\tiny m}}(Q'_Y)=\Dm$,
and so, $D_{\mbox{\tiny m}}(Q_Y)-D_{\mbox{\tiny c}}(Q'_{XY})\le 0 < \tau$,
which means that the condition $\Delta(Q_Y,R)<\tau$ holds automatically for
every $Q_{XY}$, and so, the corresponding constraint is slack.
Thus $\Pr\{\widetilde S(y^n)<e^{n\theta}\}\doteq 1$
for every feasible transmitted type.
On the other hand, for $\tau\le 0$, $D_{\mbox{\tiny m}}(Q_Y)-D_{\mbox{\tiny c}}(Q'_{XY})$
may be larger than or equal to $\tau$, making $\Delta(Q_Y,R)\ge\tau$ possible,
in which case the indicator in~\eqref{eq:T3_applied} equals $0$
and then $\Pr\{\widetilde S(y^n)<e^{n\theta}|x^n(1),y^n\}\doteq 0$.

The conclusion is that for $\tau> 0$, $\Pr\{\widetilde S(y^n)<e^{n\theta}|x^n(1),y^n\}\doteq 1$
for every transmitted type $Q_{XY}$ with $\lambda(Q_{XY},R)<\tau$.
Together with $\Pr[\mathrm{MD}|Q_{XY}]\doteq 0$ when $\lambda(Q_{XY},R)\ge\tau$),
we have $\Pr\{\mathrm{MD}\mid Q_{XY}|x^n(1),y^n\}\doteq\mathbf{1}\{\lambda(Q_{XY},R)<\tau\}$.
Summing over all transmitted types:
\begin{equation}
\beta_n(\tau,R)
\doteq
\sum_{\substack{Q:\,Q_X=P_X\\\lambda(Q_{XY},R)<\tau}} e^{-n\Dc}
\doteq
\exp\Bigl\{-n\min_{\substack{Q:\,Q_X=P_X\\\lambda(Q_{XY},R)<\tau}}\Dc\Bigr\},
\label{eq:MD_final}
\end{equation}
giving $\EMD(\tau,R)=\min_{\lambda(Q_{XY},R)<\tau}\Dc$.
However, for $\tau\le 0$ 
a type $Q_{XY}$ of $(x^n(1),y^n)$ with $\lambda(Q_{XY},R)<\tau\le 0$
contributes to $\beta_n$ at the exponential scale provided that the additional
condition, $\Delta(Q_Y,R)<\tau$, holds too.
On the other hand, if $\Delta(Q_Y,R)<\tau$,
Theorem~4.3 of \cite{MW25} yields $\Pr\{\widetilde
S(y^n)<e^{n\theta}|x^n(1),y^n\}\doteq 1$,
and this type contributes $e^{-nD_{\mbox{\tiny c}}(Q_{XY})}$ to $\beta_n$.
If $\Delta(Q_Y,R)\ge\tau$,
some bulk type of an interfering codeword $Q'_{XY}$ gives $f(Q'_{XY})\le 0$,
and Theorem~4.3 of \cite{MW25} gives $\Pr\{\widetilde
S(y^n)<e^{n\theta}|x^n(1),y^n\}\doteq 0$
(doubly exponentially small), and this type does not contribute
at the exponential scale.
Hence,
\begin{equation}
\beta_n(\tau,R)\doteq\sum_{\{Q_{Y|X}:~\lambda(Q_{XY},R)<\tau,~\Delta(Q_Y,R)<\tau\}}e^{-n\Dc},
\end{equation}
giving~\eqref{eq:MD_exp}.
This establishes~\eqref{eq:MD_exp} in both cases.

\end{document}